\documentclass[preprint,12pt,authoryear]{elsarticle}

\usepackage{amssymb}
\usepackage{amsthm}
\usepackage{graphicx}
\usepackage{epsfig}
\usepackage{epstopdf}
\usepackage{bibentry}  
\usepackage{pdflscape}
\usepackage{multirow}
\newtheorem{theorem}{Theorem}
\newtheorem{proposition}[theorem]{Proposition}

\begin{document}
%\journal{}
%\journal{Journal of Banking and Financing}
\journal{http://arxiv.org/}

\begin{frontmatter}
\title{Deleveraging, short sale constraints and market crash}
%\tnoteref{t1}
%\tnotetext[t1]{This document is a collaborative effort.}
\author{Liang Wu} %\fnref{fn1}}
\ead{liangwu@scu.edu.cn}
\author{Lei Zhang and Zhiming Fu}
\address{The School of Economics, Sichuan University,
Chengdu, PR. China, 610065}
%{\cortext[cor1]{Corresponding author}}

\begin{abstract}
In this paper, we develop a theory of market crashes resulting from a deleveraging shock. We consider two representative investors in a market holding different opinions about the public available information. The deleveraging shock forces the high confidence investors to liquidate their risky assets to pay back their margin loans. When short sales are constrained, the deleveraging shock creates a liquidity vacuum in which no trades can occur between the two representative investors until the price drop to a threshold below which low confidence investors take over the reduced demands. There are two roles short sellers could play to stabilize the market. First, short sellers provide extra supply in a bullish market so that the price of the asset is settled lower than otherwise. Second, short sellers catch the falling price earlier in the deleveraging process if they are previously allowed to hold a larger short position. We apply our model to explain the recent deleveraging crisis of the Chinese market with great success. 
 
%If such things happen, the lowering of interest rate does not help to fill the liquidity vacuum, because the interest rate is limited by the zero lower bound.  

%In the Chinese financial markets, we experience the 
\end{abstract}


\end{frontmatter}

%×Ô¶¨ÒåÂÞÂíÊý×Ö
\makeatletter
\newcommand{\rmnum}[1]{\romannumeral #1}
\newcommand{\Rmnum}[1]{\expandafter\@slowromancap\romannumeral #1@}
\makeatother
\section{Introduction}
Debt has become a hot topic since excessive debt is quoted as the key contributing factor of the economic problems in the Great Depression (\cite{fisher1933debt}), Japan's lost decade (\cite{koo2011holy}), and current financial crisis afflicting both the United States and Europe (see e.g., \cite{mian2012explains,hall2011long,guerrieri2011credit}). It has been argued by \cite{eggertsson2012debt,korinek2014liquidity} that a deleveraging shock following the excessive debt could push the economy into a situation of so called liquidity trap in which the creditors can not make up for the reduction consumption by the debtor unless the interest goes negative. The existence of liquidity trap is revealed in a simple New-Keynesian model in which two representative agents, borrowers and lenders who differ in their rates of time preference. 

Excessive debt could also be vicious to the financial system as illustrated by the recent turmoil in the Chinese stock market. How could the excessive debt and the following deleveraging shock impact the market? Does it create a liquidity trap in the financial system? To investigate these questions, we use a similar theoretic framework consisting of two representative investors. This theoretical framework is recently implemented by \cite{xu2007price} in his study of the price convexity and skewness when short sales are constrained. The model assumes that the two representative investors are low precision investors (denoting as $L$ type) and high precision investors ($H$ type). They have the same public available information, but they interpret the information differently. $H$ investors are more confident in the precision of the market signal and would buy more aggressively when the market signal is bullish than $L$ investors. When the market information is bullish, $H$ investors compete with $L$ investors and even among themselves to push the price very high. $H$ investors can borrow from the cash market if they are short of capital. In this framework, we could analyze what would happen if there is a deleveraging shock to $H$ investors who use margin loans to invest at the previous stage. The shock can come from exogenous sources, for example, it could come from the regulator to tight the monetary policy, or to clear the unmonitored leverage as recently implemented in China. It can also come from endogenous reasons, for example, the information of the market swings back from bullish to less bullish or to bearish. The price of the risky asset comes back down. When such things happen, the shrinking of $H$ investors' assets force them to liquidate their holdings to pay back their debt. 

In the real economy, one reason that the demand reduced by the debt can not be consumed by the creditor is because there is a lower bound of the interest. The other reason is that the price of the consumption product is sticky and creditor has to take all the reduced demand at the given price. In the financial system, the price of the investment product can vary freely in a short run. The deleveraging process does not create a liquidity trap. When the price comes back down, $L$ investors can make up for the demand reduced by $H$ investors at a lower price according to his price demand schedule. The reduced demand of the investment product by the debtor can be consumed completely as long as the price is low enough. 

It is a relief that there is no liquidity trap problem in the financial market, however, the deleveraging process would cause the crash of the market when investors are under short sale constraints. If there are no short sale constraints, there are two roles short seller could play to stabilize the market. First, when the information is very bullish, as $H$ investors push the price high, $L$ investors can provide extra supply to $H$ investors by short selling, so that the price is settled lower than otherwise. Second, since $L$ investors short sell following the price all the way up, when the market swings back down, the short sellers immediately start to buy back the asset all the way down until the deleveraging process finishes. 
If there are short sale constraints, $L$ investors are crowded out before $H$ investors push the price further high among themselves to compete for the demand. When the price comes back down, $L$ investors would not buy untill the price drops to the level right before they leave the market. The market could experience a free fall (we call it liquidity vacuum in contrast with the term liquidity trap) during which $H$ investors can not liquidate their holdings to $L$ investors to reduce their debt. 
It is ironic that short sale constraints become the reason of the market crash, while the implementation of the short sale constraints is supposed to protect the market from crash by restricting short selling activity. 

The literature on the impacts of short sale constrains is extensive. It has been revealed by \cite{miller1977risk,jones2002short} that short sale constraints can prevent negative information from being expressed and therefore lead to the overpricing of securities. \cite{saffi2010price} study more than 12,600 stocks from 26 countries and conclude that relaxing short-sales constraints has improved the price efficiency without an association with an increase in either price instability or the occurrence of extreme negative returns. \cite{bris2007efficiency,chang2007short,chang2014short,zhao2014short} make similar empirical observations that market efficiency can be improved when the short sale constraints are lifted. 

A mechanism that short sale constraints could cause the market crash has been studied by \cite{hong2003differences}. In their model, the market price is set by the rational expectation to reflect the information revealed by the trading activities of heterogeneous investors who hold different opinions about the stock's price. When investors are under short sale constraints, bad information is not fully reflected in a bullish market. However, when the market swings back, hidden bad information and thus a larger scale of price adjustment comes out and potential could cause the crash of the market. In this paper, we offer a different mechanism for the market to crash when short sale constraints are present. The mechanism lies in the deleveraging shock for debtor to unwinding their loans. The short sale constraints create a liquidity vacuum in which $L$ investors can not make up for the reduced demand by $H$ investors untill the price crashes to a low level. 

Short sale constraints are still fully imposed in some emerging markets (\cite{bris2007efficiency}). \cite{bai2007asset} list various costs and legal and institutional restrictions that could impose constraints on short selling. During the turbulent times, short sales are often partially re-imposed although they have been fully practiced under norm circumstances. \cite{qiu2013effect} argues that short-sale constraints will continue to affect the global financial markets in the future. It is therefore necessary for investors and especially policy makers to understand the potential problems including the market crash risk associated with short sales are constrained.  

The remainder of the article is organized as follows. In section 2 we solve the model and analyze the consequence when there is a develeraging shock.
In section 3, we first present numerical examples to show that market crashes could be caused by a deleveraging shock when short sales are constrained. Then we apply our model to the case of the most recent crisis in the Chines stock markets. Section 4 concludes.

%% main text
\section{The Model}
We extend \cite{xu2007price}'s model for the inclusion of a deleveraging shock. The market consists of two assets, namely, a risky asset with a random payoff equal to $x$ and a risk-free asset with gross interest rate $R$. There are two types of investors who have the same public information, but disagree with the precision of the information. When the signal $s$ is revealed to the public at time 1, trading happens between the two types of investors and the market clears at equilibrium. At time 2, there is a sudden deleveraging shock. Debtors are required to unwind their loans if they borrow from the cash market at time 1.

We use the same notations as in \cite{xu2007price}. Assume that the public information is the signal with respect to the payoff of the risky asset at time 1, denoted as $s=x+\epsilon$. $x$ and $\epsilon$ are independent normal distributed random variables, i.e., $x\sim\mathcal{N}(\mu_x,\tau^{-1}_x)$ and $\epsilon\sim\mathcal{N}(0,\tau^{-1}_\epsilon)$, where $\tau_x$ and $\tau_\epsilon$ are the precision of $x$ and $\epsilon$, respectively. Assume that investors agree over the distribution of $x$ but disagree over the prevision $\tau_\epsilon$. Low confidence investors believe that the precision the $\epsilon$ is $\tau^2_L$, and high confidence investors who believe that the prevision is $\tau_H$, $\tau_L<\tau_H$ \footnote{The precision of investors hold about the signal $s$ can have multiple values and can even be continuous. If we divide the investors into two groups according to if their precision is lower or higher than a threshold $T$, type $\tau_H$ is the representative investor of the high precision investors above $T$ such that $\tau_H = \frac{\int_T^\infty d\lambda_\tau \tau }{\int_T^\infty d\lambda_\tau}$, and type $\tau_L$ is the representative investor of the low precision investors below $T$ such that $\tau_L = \frac{\int_0^T d\lambda_\tau \tau }{\int_0^T d\lambda_\tau}$.}. The proportion of $\tau_H$ investors is $\lambda$.

%\subsection{Equilibrium when there is no short sale constraints}

%First, the posterior precision and mean of type $\theta$($\theta=L,H$) are, respectively,
%\begin{eqnarray}
%&\textnormal{Var}_\theta(x|s) \equiv \hat{\tau}_\theta(x|s) = \tau_x + \tau_\theta
%&\textnormal{E}_\theta[x|s] \equiv \hat{\mu}_\theta = %\frac{\tau_x}{\hat{\tau}_\theta}\mu_x + \frac{\tau_\theta}{\hat{\tau}_\theta}s \\
%\end{eqnarray}
Investors maximize expected utility of wealth at the end of each period, whose utility function is assumed to be\footnote{The mean-variance utility function is equivalent to the utility function $U(W_\theta(t))=-e^{-aW_\theta(t)}$ under the normality assumption of $s$.},
\begin{eqnarray}
\label{eq:utility}
\textnormal{max}_{y_{1,\theta}} U(W_{1,\theta}) \equiv \textnormal{E}_\theta[W_{1,\theta}|s]-\frac{a}{2}\textnormal{Var}(W_{1,\theta}|s) \\ \nonumber
\textnormal{s.t.} \quad W_{1,\theta} = (W_{0,\theta}-y_{1,\theta}P_1)R + y_{1,\theta}x
\end{eqnarray}
where $W_{0,\theta}$ is the total initial wealth of type $\theta$ investor.  

After take a derivative over Eq.~(\ref{eq:utility}), we have a price and demand relationship for each type of investors to follow if no constraints are imposed, 
\begin{eqnarray}
\label{eq:mid}
y_{1,\theta} = \frac{(\hat{\mu}_\theta-RP_1)\hat{\tau}_\theta}{a} 
\end{eqnarray}
where $\hat{\tau}_\theta\equiv \textnormal{Var}^{-1}_\theta(x|s) = \tau_x + \tau_\theta$ is the posterior precision, and $\hat{\mu}_\theta \equiv \textnormal{E}_\theta[x|s] =  \frac{\tau_x}{\hat{\tau}_\theta}\mu_x + \frac{\tau_\theta}{\hat{\tau}_\theta}s$ is the posterior mean of the actual payoff $x$ given the public signal of $s$.

\subsection{Short sale constraints}
Different from Xu(2007)'s modeling, we assume that the short sale constrains are not fully imposed. There is a cap on the proportion of the asset that can be borrowed short sale practice. This modification enables us to apply the model to analyze the Chinese stock market where the margin lending and short sale policy has been practiced for more than 5 years. The scale of the short sale, however, is still very small. The other reason for this modification is that we can see that the gradual lift of the short sale constraints can lessen the risk of market crash. 

Assume that the proportion of the assets investors can borrow for short sales does not exceed $N$, i.e., $\lambda y_L \ge -N, (1-\lambda)y_H \ge -N$, where $N\ge 0$. The equilibrium solution is given by the following proposition:
\begin{proposition}
\label{prop:shortsaleBorrowing}
(a) If condition: $\hat{\mu}_H-\hat{\mu}_L>\frac{a(1+N)}{(1-\lambda)\hat{\tau}_H}+\frac{aN}{\lambda \hat{\tau}_L}$ (or $s>\mu_x + \frac{a(1+N)\hat{\tau}_L}{(1-\lambda)(\tau_H-\tau_L)\tau_x}+\frac{aN\hat{\tau}_H}{\lambda(\tau_H-\tau_L)\tau_x}$), 
the equilibrium is given as $P_1=P_H, \lambda y_{1,L} = -N,(1-\lambda)y_{1,H} = 1+N$;

(b) If condition: $\hat{\mu}_L-\hat{\mu}_H>\frac{a(1+N)}{\lambda\hat{\tau}_L}+\frac{aN}{(1-\lambda)\hat{\tau}_H}$ (or $s< \mu_x - \frac{a(1+N)\hat{\tau}_H}{\lambda(\tau_H-\tau_L)\tau_x}-\frac{aN\hat{\tau}_L}{(1-\lambda)(\tau_H-\tau_L)\tau_x}$), the equilibrium is given as $P_1=P_L$, $\lambda y_{1,L} = 1+N, (1-\lambda)y_{1,H} = -N$;

(c) If condition (a) and (b) do not hold, the equilibrium is given as,
\begin{eqnarray}
P_{1} = \frac{\lambda\hat{\tau}_L\hat{\mu}_L+(1-\lambda)\hat{\tau}_H\hat{\mu}_H-a}{R(\lambda\hat{\tau}_L+(1-\lambda)\hat{\tau}_H)} \\
y_{1,L} = \frac{1/a(1-\lambda)\hat{\tau}_L\hat{\tau}_H(\hat{\mu}_L-\hat{\mu}_H)+\hat{\tau}_L}{\lambda\hat{\tau}_L+(1-\lambda)\hat{\tau}_H} \\
y_{1,H}=\frac{1/a\lambda\hat{\tau}_L\hat{\tau}_H(\hat{\mu}_H-\hat{\mu}_L)+\hat{\tau}_H}{\lambda\hat{\tau}_L+(1-\lambda)\hat{\tau}_H}
\end{eqnarray}

In the above expressions, $P_H$ and $P_L$ are the price corresponding to the demand $(1-\lambda)y_{1,H} = 1+N$ of $H$ investors,  and $\lambda y_{1,L} = 1+N$ of $L$ investors, respectively. They are given by the following expressions:
\begin{eqnarray}
\label{eq:ph}
&P_H = \frac{1}{R}(\hat{\mu}_H-\frac{a(1+N)}{(1-\lambda)\hat{\tau}_H}) \\
\label{eq:pl}
&P_L = \frac{1}{R}(\hat{\mu}_L-\frac{a(1+N)}{\lambda\hat{\tau}_L})
\end{eqnarray}
\end{proposition}

\begin{proof}
Under (a), according to the price demand relationship of Eq.~(\ref{eq:mid}), the demand $\lambda y_{1,L}$ of $L$ investors exceeds the short sale constraints. The short sales of $L$ investors are saturated at the lower bound $\lambda y_{1,L}=-N$, the equilibrium price is therefore determined by the demand of the $H$ investors from the market clearing condition $(1-\lambda)y_{1,H} = 1-\lambda y_{1,L} = 1+N$. From the price demand relationship, we have $P_1 = P_H$. Similar reasoning is applied to (b). In this case, the short sales of $H$ investors are saturated at $-N$. $P_1$ is determined by $\lambda y_{1,L} = 1+N$. Under (c), the equilibrium is solved by plugging Eq.~(\ref{eq:mid}) for both $L$ and $H$ investors to the market clearing condition $\lambda y_{1,L,0} + (1-\lambda)y_{1,H,0} = 1$.
\end{proof}

\subsection{Deleveraging shock}
Now we consider a scenario when $s$ is very bullish as in condition (a) of Proposition~\ref{prop:shortsaleBorrowing}. Only $H$ investors hold the risky asset and would like to leverage if they are short of capital. If $H$ investors are in debt, they have to liquidate part or all of the risky asset when facing a deleveraging shock. Since in this case $L$ investors are crowded out, the price $P_1$ of the risky asset at time 1 is higher than the price threshold $P_s$ below which $L$ investors are willing to absorb the demand reduced by the $H$ investors. The threshold price $P_s$ is determined if we set the demand of $L$ investors to be $\lambda y_{1,L}=-N$,
\begin{eqnarray}
\label{eq:ps}
P_s = \frac{1}{R}(\hat{\mu}_L+\frac{aN}{\lambda\hat{\tau}_L}).
\end{eqnarray}

If under condition (a) of Proposition~\ref{prop:shortsaleBorrowing}, 
\begin{eqnarray}
P_1 - P_s = P_H - P_s >0
\end{eqnarray}
The positive difference $P_1-P_s$ between $P_1$ and $P_s$ creates a liquidity vacuum. The price of the risky asset would experience a free fall from $P_1$ to $P_s$ before the liquidation of $H$ investors can actually start. There is no buying from $L$ investors to help $H$ investors unwind their debt. 
Note from Eq.~(\ref{eq:ph}) and Eq.~(\ref{eq:ps}), $P_H$ is a decreasing function of $N$ and $P_s$ is an increasing function of $N$. If the short sale constraints of $L$ investors are less strict ($N$ is larger), short sales from $L$ investors can reduce the impact of the deleveraging shock by reducing the difference $P_1-P_s$ from two sources. First, at time 1 they provide extra liquidity to $H$ investors by short selling when $s$ is very bullish, so that $P_1$ can not be pushed that high. Second, the short sellers can sell all the way up to the equilibrium price if they are under less strict constraints. At time 2 they step in immediately at a high price level to absorb the demand reduced by $H$ investors when there is a deleveraging shock. On the contrary, if the short sellers are force to leave the market before the price is pushed high due to the strong demand of $H$ investors, they could come back only when the price drops back to the level they leave the market. 

In addition to the free fall of the price before $L$ investors come back, the price needs to adjust further for $H$ investors to sell off enough asset to pay back their debt. Assume that the leveraging ratio is $\eta$, i.e., the debt $H$ investors owe is $\eta P_1$ ($P_1$ is also the market value of the risky asset when there is only 1 share). 

Let $P_2$ be the price at which the risk asset finally stabilizes, 
\begin{eqnarray}
\nonumber
\int_{P^*}^{P_2}P\lambda d y_{1,L} = \int_{P^*}^{P_2}P\frac{-\lambda\hat{\tau}_LR}{a}dP \\
=\frac{\lambda\hat{\tau}_LR}{2a}(P^{*2}-P_x^2) = \eta P_1 \\
P_2^2 = P^{*2}-\frac{2a\eta P_1}{\lambda\hat{\tau}_LR} 
\end{eqnarray}
where $P^*=\textnormal{min}(P_1,P_s)$ is the highest price $L$ investors start to add their holdings of the risky asset. It is possible that $H$ investors can not pay back all debt when $\eta$ is large enough. In this case, the lowest level $P_2$ drops to is after $H$ investors dump all of their holdings, and thus $P_2$ is bounded below by $P_2\ge \frac{1}{R}(\hat{\mu}_L-\frac{a}{\lambda\hat{\tau}_L})$ when $y_{2,L}=1$.

In the above analysis, we assume that $s$ remains the same during the deleveraging process. The crash of the market conveys a negative signal. $s$ could decrease significantly, which drags the price further deeper and crates an even bigger problem.

Can the lowering of interest stabilize the market? First, the lowering of interest to support the crash of a financial bubble is not justified if the financial crisis does not hurt the real economy. Even if it does, there is a limit the lowering of the interest can increase $P_s$ since the interest rate is lower bounded by $(\hat{\mu}_L+\frac{aN}{\lambda\hat{\tau}_L})$. If $s$ is large enough, the gap between $P_1$ and $P_s$ still exists,

\section{Results}

\subsection{Simulation}
In this section, we present some numerical examples to illustrate how the market crashes when there is a deleveraging shock. The parameters in our examples are set as $\mu_x = 1.5, \tau_x = 1, R=1.05, \tau_L = 0.5, \tau_h = 1.5, \lambda = 0.5, a = 0.5, s = 5$. We set the short sale constraints parameter $N$ in a rather large region $N\in [0,0.6]$ to cover many different scenarios. In some case $N\in [0,0.5)$, we have $s>\mu_x+A$ and therefore $P_1>P_s$. There are some other cases $N\in (0.5,0.6]$, we have $\mu_x-B<s< \mu_x+A$, and therefore $P_1<P_s$. Suppose at time 1, $H$ investors have enough cash (they can borrow if not) to clear the market at $P=P_H$ when $s>\mu_x+A$, and at $P=P_{1,0}$ when $\mu_x-B<s< \mu_x+A$. The leverage ratio is given as $\eta\in[0,0.4]$. The amount of debt $H$ investors owe at time 1 is equal to $\eta P_1$.

First, we plot $P_1$ and $P_s$ against $N$ in Fig.~\ref{fig:price}. As expected by the theoretic model, the figure shows that $P_1$ and $P_s$ are respectively the decreasing function and increasing function of $N$ in a bullish market. In addition, we also plot the threshold price $P_{s0}$ when we set $R=1$ with zero actual/nominal interest rate. It shows that lowering the interest rate gives $L$ investors incentive to step in at a higher price but can not completely solve the liquidity vacuum problem. At $N=0.5$, $P_1$ and $P_s$ meet and the liquidity vacuum during which the price would free fall vanishes before $H$ investors are able to unwind the debt. The final price return $\log(P_2)-\log(P_1)$ is reported in Fig.~\ref{fig:return} for different combinations of $\eta$ and $N$. The top figure presents the three-dimensional view of the log-return as the function of $\eta$ and $N$. The bottom figure presents the contour of the log-return in the $\eta-N$ plane. The maximum reduction of the price (log return is $-0.32$) is obtained at $(\eta,N)=(0.4,0)$, when leverage is at the highest level and no short sale is allowed. The minimum reduction of the price (log return is $0$) is obtained at $(\eta,N)=(0,[0.5,0.6])$, when short sale is allowed to be more than 0.5 share and there is no debt. Notice that even when there is slight debt and when short sales are not fully practiced ($\eta=0^+, N<0.5$), the market would experience a non-zero free fall crash. 

When $N\in [0,0.5)$, the larger the leverage ratio $\eta$ is or the smaller the short sale constraint ratio $N$ is, the deeper the market would crashes into. When $N\in (0.5,0.6]$, further increase of $N$ does not change the price return, since in this case $P_1<P_s$. We have already eliminated liquidity vacuum problem. The drop of the price is purely caused by the liquidating operation of $H$ investors. 

\begin{figure}
\centering
\includegraphics[width=0.9\linewidth]{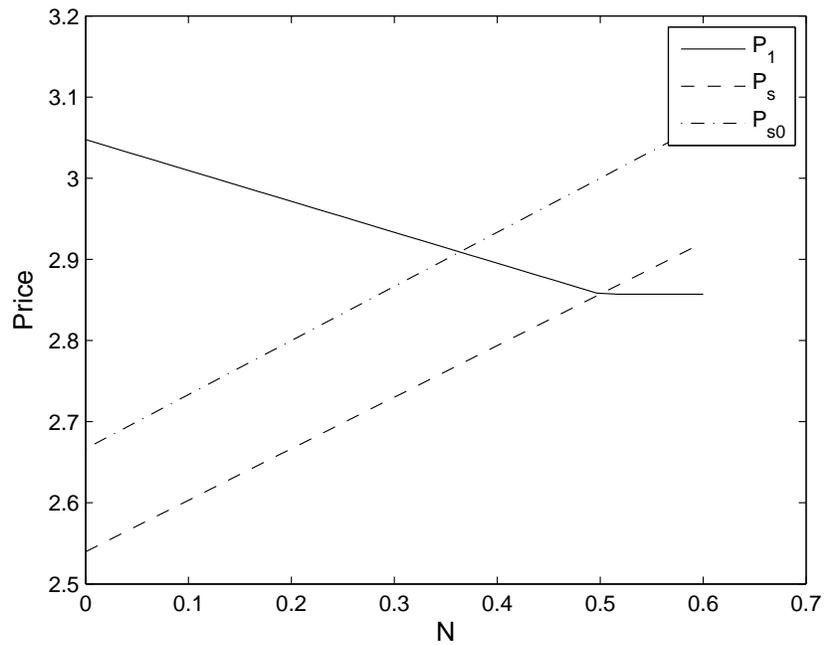}
\caption{$P_1$ and $P_s$ are respectively the decreasing function and increasing function of $N$ in a bullish market $s=5$. $P_1$ is the equilibrium price at time 1, and $P_s$ is the threshold price below which $L$ investors are willing to take over the reduced demand by $H$ investors. $P_{s0}$ is the the threshold price when we set $R=1$ (zero actual/nominal interest rate) at time 2. The short sale constraint ratio $N$ is the maximum proportion of the security allowed for short sales. $P_1-P_s$ is the gap the price would free fall in the presence of a deleveraging shock.}
\label{fig:price}
\end{figure}

\begin{figure}
\centering
\includegraphics[width=0.9\linewidth]{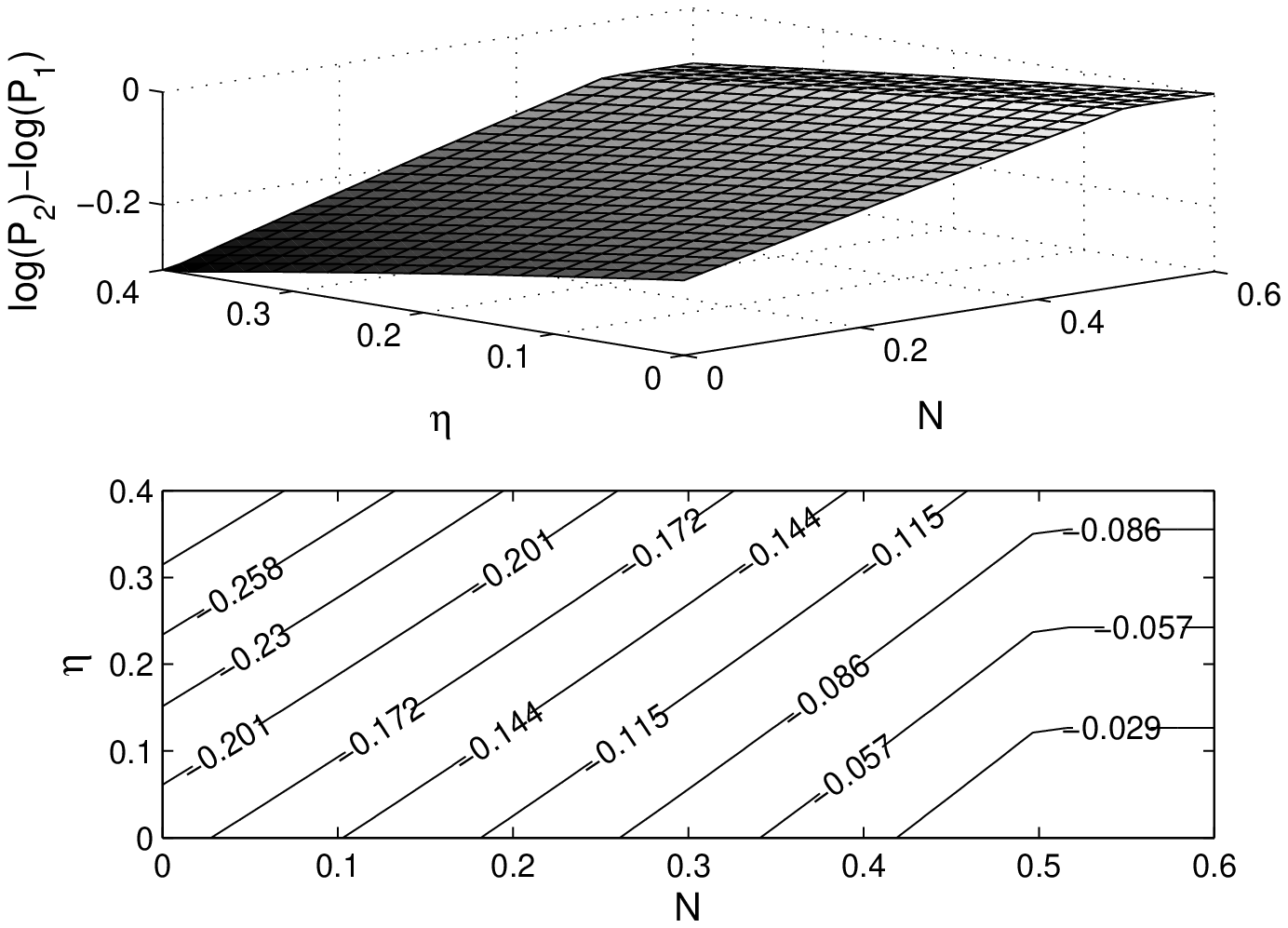}
\caption{Log-return of the price from time 1 to time 2 ($\log(P_2)-\log(P_1)$) after $H$ investors liquidate part or all of their holdings to pay back their debt. The top figure is the 3-d view of the log-return as the function of the leverage ratio $\eta$ and short sale constraint parameter $N$. The bottom figure is the 2-d contour version of the top figure.}
\label{fig:return}
\end{figure}

\subsection{Empirical Results}
The Chinese market has experienced an unseen rally for almost 10 years from the late half 2014 to the early half of 2015. The CSI300 index, an index of China’s biggest-listed companies, reaches the highest point of 5335 on June 12, 2015 from 2164 on July 1, 2014. It is a massive 147\% surge in less than 1 year. But it soon crashes to 3663 after 17 trading days on July 8, 2015. 
The maximum daily loss of CSI300 index is 8.25\% recorded on June 26, 2015 (even when there is a daily 10\% limit of price change for all stocks traded in the Chinese markets). There are 8 days out of 17 days from June 12, 2015 to July 8, 2015, when more than 900 stocks halt trading because they hit the 10\% daily bottom limit. The stock prices have been pushed so quickly in less than 1 year. It is firmly in the bubble territory and should reasonably be followed by a correction. However, the anticipated correction of over-valued stocks hardly seems cause for the dramatic disaster. The government has taken a spectacle of drastic actions to save the market. The People's Bank of China intervenes on June 27, 2015 by cutting the interest by 0.25\%. The initial public offerings are suspended on July 3, 2015 by limiting the supply of shares to drive up the prices of those already listed. Far from saving the market from drowning, those efforts only pushed the market further under water. Starting from July 6, the China Securities Finance Corporation Limited (CSF) starts to buy blue chip stocks and even small and medium-sized stocks. On July 8, CSF lends 260 billion yuan to 21 brokerage firms so they can purchase "blue chip" stocks on top of what the 120 billion the brokerages vowed to buy. The market finally stabilizes and rebounces on July 9, 2015. 
Fortunately, the financial crisis in China is not a systemic risk. Less than 15\% of household financial assets are invested in the stock market: which is why soaring shares did little to boost consumption and crashing prices will do little to hurt it.
 
One fact can not be ignored to understand the frenetic behavior of the market. The rally of the Chinese market is boosted by the margin lending. Many stocks were bought on debt. People borrow money with an interest rate of about annual 8.6\% to invest in the market from the margin lending channel offered by brokers for a designated list of stocks. From the public monitored data of margin lending, the outstanding margin loan amounts more than 2,200 billion yuan, which is about 11\% of the market value of those designated stocks. In addition to the public available data, there are other sources of margin lending which is not monitored by the China's Securities Regulatory Commission. The unmonitored margin lending loan amounts more than 1,700 billion yuan with a cost of 13\%-20\% annul rate. It is not sustainable and the unwinding of these debt causes the rapid crash of the market which the government has been unable to stop.

In this section, we use the monitored margin lending and short selling data to test our model. In the Chinese market, a list of designated stocks are allowed for margin lending and short sale practice. Up to June 12, 2015, there are 895 of them, among which 846 stocks have outstanding margin lending loans. We use the outstanding margin lending loans and short sale data recorded on June 12, 2015 when the CSI300 index hits the highest point. The price return is calculated from June 12, 2015 to July 8, 2015 when the market hits the bottom for the first stage. During this period of fast melting down process, although the People's bank of China cuts the interest rate and IPO is suspended. Only until the rather late part of the considered period (starting from June 6th, 2015), there is direct massive intervention from CSF by buying falling not only blue chips but also medium or small-sized stocks. Th considered melting down period can be viewed as purely dragged by the market force.

The descriptive statistics of the variables for the considered period is list in Table.~\ref{tab:my_labelx}. $\beta$s of the considered stocks are estimated using data from January 1, 2013 to June 12, 2015. The average log-return of the considered stock is -0.59 with standard deviation of 0.31. Since the stocks in the designated list are mostly large-sized, the average beta is only 0.74. The leverage ratio and short sale ratio are both calculated by dividing the outstanding lending loan (for leverage ratio) and short interest (for short sale ratio) by the outstanding market value. The average leverage ratio is at the high level of 0.11. Although the margin lending and short sale policy has been implemented back from 2010, there are still limited sources for investors to borrow the stocks for short sales. The short sale ratio is rather small compared to the leverage ratio. The average short sale ratio is only $1.9\times 10^{-4}$.

\begin{figure}
\centering
\includegraphics[width=0.9\linewidth]{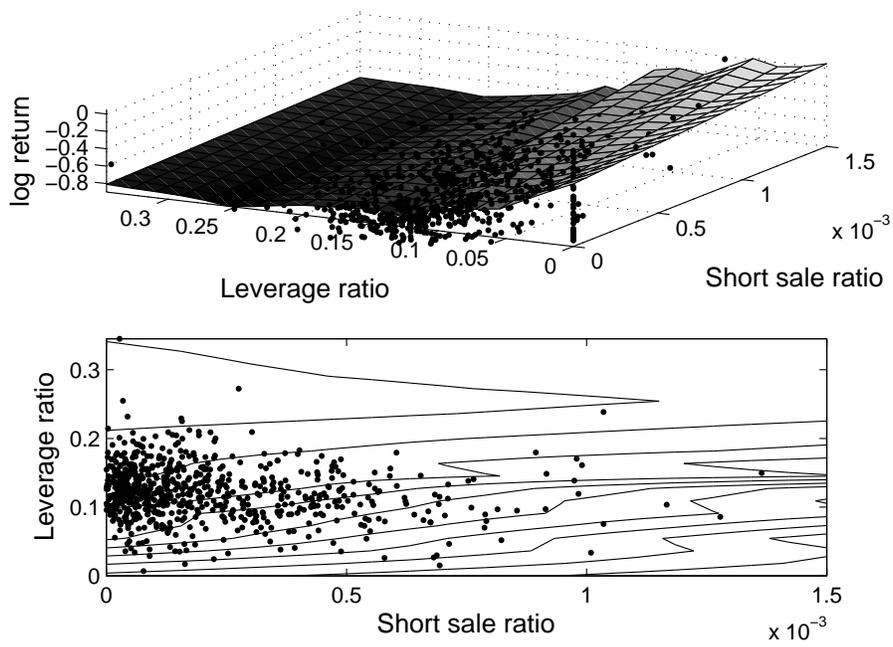}
\caption{Log-returns of 846 Chinese stocks in the designated list for margin lending and short sale practice. The dots are the actual data. The surface in the top figure and the contour in the bottom are the interpolations of the real data using a local linear regression.}
\label{fig:act}
\end{figure}

\begin{table}[]
\centering
\caption{Description statistics of 846 Chinese stocks that have non-zero outstanding margin loans on June 12, 2015. The time period to measure the log-return is from June 12, 2015 to July 8, 2015. The outstanding margin loans and short sale data used to calculate the leverage ratio and short sale ratio are recorded on June 12, 2015. The leverage ratio and short sale ratio are calculated using the published outstanding margin loans and short interest with respect to the outstanding market value of each 
considered stock.}
\begin{tabular}{c|c|c}
\hline
 & \textnormal{}{Mean} & \textnormal{Std.} \\
 \hline
\textnormal{log-return} & -0.59  & 0.31   \\
\textnormal{$\beta$} & {0.74} & {0.15} \\
\textnormal{smb} & {0.33} & {0.21} \\
\textnormal{hml} & {0.029} & 0.24 \\
\textnormal{leverage ratio} & {0.11} & {0.053}\\
\textnormal{short sale ratio} & {$1.90\times 10^{-4}$} & {$2.51\times 10^{-4}$} \\
\hline
\end{tabular}
\label{tab:my_labelx}
\end{table}

\begin{table}[]
\centering
\caption{Regression analysis of stock returns. The explained variables are the returns of Chinese stocks from June 12, 2015 to July 8, 2015. The three columns report the regression results when the dependent variables are: the FAMA three factors; the leverage ratio and short sale ratio; the FAMA three factors, the leverage ratio and short sale ratio, respectively. The leverage ratio and short sale ratio alone can give an explanation of 86\%, larger than the FAMA three factor model.}
\begin{tabular}{c|c|c|c}
\hline
$\beta$ & -0.41 (-5.88)*** &  &-0.18 (-2.99)** \\
\textnormal{smb}& -0.40 (-6.91)*** &  & -0.31 (-5.91)*** \\
\textnormal{hml}& -0.079 (-1.65)* &   & 0.038 (0.93) \\
\textnormal{leverage ratio} & & -3.14 (-19.31)*** & -2.91 (-18.00)*** \\
\textnormal{short sale ratio} &  & 171.20 (5.00)*** & 83.91 (2.46)**   \\
\textnormal{const} & -0.15 (-3.10)*** & -0.28 (-13.89)*** & -0.054 (-1.26) \\
\hline
\textnormal{F-Statistics} & 1000 & 1731 & 976  \\
$R^2$ adjusted & 0.825 & 0.860 & 0.874   \\
\textnormal{\# of samples} & 846 & 846 & 846 \\
\hline
\end{tabular}
\label{tab:my_label2}

\small{\emph{Note: One asterisk (*), two asterisks (**), and triple asterisks (***) denote significance at 10\%, 5\%, and 1\% critical level, respectively.}}
\end{table}

First, we show the log-return data of the considered 846 stocks. The data is presented in a similar way as the numerical examples. The log-returns of the stocks from June 12, 2015 to July 8, 2015 are plotted against the leverage ratio and short sale ratio. The results are depicted in Fig.~\ref{fig:act}. The dots are the actual data. The surface of the top figure and the contour in the bottom are the interpolated results using a local linear regression to fit the samples (using the "fit" function from Matlab the fitting method "lowessfit"). The results are qualitatively similar to the numerical results. The log-returns are clearly a decreasing function of leverage ratio and increasing function of short sale ratio. There are some quantitative difference between the actual results and the numerical examples. For example, the minimum log-return is about -0.8, much lower than the theoretic prediction. The log-return is very sensitive to the short sale ratio. Even the variation of short sale is in the order $10^{-3}$, it changes the log-return significantly. The purpose of this study is to qualitatively understand the major cause of the market crash and its implication for policy makers. Fitting perfectly to the actual data is not the purpose of this study. Besides, there are many other factors that could influence the observations. As mentioned in the model section, a deleveraging shock conveys a negative signal to the market. $s$ could decrease significantly, which drags the price further deep.

We then use a linear regression analysis to see quantitatively how much the log returns can be explained by the major factors considered in this study, i.e., the leverage ratio and short sale ratio. The results are reported in Table.~\ref{tab:my_label2}. The first equation in the column 1 is the FAMA three factor model. The coefficient of $\beta$ is negative because the market return is negative during the considered time period. The second column reports the results when we regress the log-return with the leverage ratio and short sale ratio. As expected, we have both statistically and economically significantly negative and positive contributions from the leverage ratio and the short sale ratio, respectively. We can see that from the $R^2$ the leverage ratio and short sale ratio can explain the 86\% of the variance of the data, higher than the FAMA three factor model. When all variables are present in the regression, $R^2$ only  increases marginally from 0.86 to 0.874. The sign of the leverage ratio and short sale ratio are again as expected. %One might notice that the coefficient of the short sale ratio is less significant than that of the leverage ratio since the data range for short sale is very limited, most of which are concentrated near 0. 
By looking at the constant term of the third column, which is the remainder (only -0.054 out of -0.59) of the aggregate market return that is not explained by our model, we see that our model can explain not only the cross sectional variation of the price return up to 87\%, but also the 91.5\% (-0.54 out of -0.59) of the aggregate market return.

\section{Conclusion and Discussion}
In this study, we implement a theoretical model consisting of two representative investors to investigate the problems of excessive leverage. We offer a mechanism that excessive leverage could cause the crash of the market if the short sales are constrained. There are two roles the short sellers can play in stabilizing the market in a bullish environment where over-confident investors take excessive leverage to invest. The short sellers provide extra supply by short selling to temper the exuberant market so that the over-confident investors would not push the market too high. During the deleveraging process when the market swings back, the short sellers can catch the falling price earlier to close their short positions. If the short sales are constrained, the short sellers are crowded out before the price is pushed too high by the over-confident investors. It creates a liquidity vacuum in which short sellers would not buy the reduced demand of over-confident investors untill the price free fall to the level the short sellers are previously crowded out.  

Excessive leverage is vicious but should not be blamed. There are no effective ways to stop the loans if investors believe it is profitable to invest in the risky asset by borrowing money from a cash market. It is equally hard to prevent people from financing through other channels for stocks that are not in the margin lending list. What could be done is to introduce a balancing force from the short sellers so that the market can be settled at a fair price. 

Short sales are blamed to cause the crash of the market. People forget that it is the short sale constraints that push the price of the market to the bubble territory in the first place. Our study shows that short sellers can even help stabilize the market by absorbing the reduced demand by the debtors when they unwind their loans dealing a deleveraging shock. It is important for the investors and especially the policy makers to understand the roles short sellers can play in a health market. The short sales have been practiced for more than 5 years in China. But the scale of short sales is still very small, with a short interest on the order of $10^{-4}$ of total market capitalization. But the short interest in the U.S. market have already reached the level over $1.7\%$ of total market capitalization as shown by 2001\cite{d2002market}. Sources need to expand for investors to borrow stocks for short sales.

There are three index future contracts traded in the Chinese markets, including two recently launched on April 16, 2015 before the crisis. The spot markets are constrained in short selling. The future market offer investors an alternative to short the market and a tool to manage their risk in a downward-trending market. During the deleveraging crisis, the index future contracts were traded at a maximum discount rate of over 10\% with respect to their spot indices. People may argue that the short sales on the future market drag the spot market into the crisis. We believe it is more likely that people turn to sell on the future markets when they can not close their long position in stocks when they hit the 10\% limit. In our study, the major part of the aggregate market return (-0.54 out of -0.59) can be explained by the 5 dependent variables (the FAMA three factors, the leverage ratio and short sale ratio). It is left for the future study on the trading activities on the future markets to find out whether the aggregate crash of the market is further caused by the short sale on future markets or the crash of the future markets are simply the mirror the crash of the spot market.

\section{References}
%\begin{footnotesize}
\bibliographystyle{elsarticle-harv}
\bibliography{second}
%\end{footnotesize}

\end{document}